\pdfoutput=1
\def\year{2018}\relax
\documentclass[letterpaper]{article} 
\usepackage{aaai18}  
\usepackage{times}  
\usepackage{helvet}  
\usepackage{courier}  
\usepackage{url}  
\usepackage{graphicx}  
\usepackage[usenames, dvipsnames]{color}
\usepackage[inline]{enumitem}
\usepackage{times}
\usepackage{latexsym}
\usepackage{latexsym}
\usepackage{amsmath}
\usepackage{amssymb}
\usepackage{amsthm}
\usepackage{latexsym}
\usepackage{amsmath}
\usepackage{amssymb}
\usepackage{amsthm}

\usepackage{xspace,color,upgreek} 
\usepackage{graphicx}
\usepackage{tikz,pgf}
  \usetikzlibrary{automata,positioning,matrix,calc,petri,arrows}

\theoremstyle{plain}
\newtheorem{theorem}{Theorem}
\newtheorem{lemma}[theorem]{Lemma}
\newtheorem{remark}{Remark}

\newtheorem{example}{Example}
\newtheorem{definition}{Definition}
\newtheorem{corollary}[theorem]{Corollary}

\def\it{\begin{itemize} }
\def\-{\item }
\def\ti{\end{itemize} }
\def\en{\begin{enumerate} }
\def\ne{\end{enumerate} }

\def\exptime{\textsc{exptime}\xspace}

\def\ptime{\textsc{ptime}\xspace}

\DeclareMathOperator{\nextX}{\mathsf{X}}

\DeclareMathOperator{\until}{\mathbin{\mathsf{U}}}

\DeclareMathOperator{\always}{\mathsf{G}}

\DeclareMathOperator{\eventually}{\ensuremath{\mathsf{F}}\xspace}

\newcommand{\true}{\mathsf{true}}
\newcommand{\false}{\mathsf{false}}

\newcommand{\CTLS}{\ensuremath{\mathsf{CTL}^*}\xspace}

\def\int{\mathbb{Z}}

\newcommand{\Math}[1]{\ensuremath{#1}}

\newcommand{\modecal}[1]{{\Math{\mathcal{#1}}}}

\newcommand{\A}{\modecal{A}} 
 
\newcommand{\E}{\modecal{E}}

\def\planeaux!#1:#2<-#3!{\Math{#1 \mbox{\rm:} #2\; \leftarrow #3}}

\def\planeaux!#1<-#2!{\Math{#1 \leftarrow #2}}

\long\def\eatpar#1{%
\ifx#1\par                      
\let\nextmove=\eatpar           
\else
\let\nextmove=#1
\fi
\nextmove
}

\def\qed{\hfill{\qedboxempty}      
  \ifdim\lastskip<\medskipamount \removelastskip\penalty55\medskip\fi}

\def\qedboxempty{\vbox{\hrule\hbox{\vrule\kern3pt
                 \vbox{\kern3pt\kern3pt}\kern3pt\vrule}\hrule}}

\def\qedfull{\hfill{\qedboxfull}   
  \ifdim\lastskip<\medskipamount \removelastskip\penalty55\medskip\fi}

\def\qedboxfull{\vrule height 4pt width 4pt depth 0pt}

\newcounter{bean}

\newenvironment{tightenumerate}{
                \begin{list}{
                  {\mbox {
                      \arabic{bean}.\/}}}{\usecounter{bean}
                      \setlength{\itemsep}{-1pt}\setlength{\topsep}{0pt}}}{
                \end{list}}

\newenvironment{tightitemize}{
                \begin{list}{$\bullet$}{
                    \setlength{\itemsep}{-1pt}}{\setlength{\topsep}{0pt}}}{
                \end{list}}

\newcommand{\under}[1]{\mbox{\underline{\it\smash{#1}\vphantom{\lower.05ex\hbox{
x}}}}}

\newcommand{\commentarea}[1]{}

\newcommand{\limp}{\supset}

\usepackage{hyperref}

\allowdisplaybreaks

\def\ag{\ensuremath{\mathsf {ag}}\xspace}
\def\env{\ensuremath{\mathsf {env}}\xspace}

\def\AP{{\sf{Var}}}

\def\sat{\vartriangleright}

\newcommand{\DFA}{\text{{DFA}}\xspace}

\newcommand{\DPW}{\text{{DPW}}\xspace}

\newcommand{\LTLf}{\text{{LTLf}}\xspace}
\newcommand{\LTL}{\text{{LTL}}\xspace}
\newcommand{\LDLf}{\text{{LDLf}}\xspace}
\newcommand{\LDL}{\text{{LDL}}\xspace}

\newcommand{\SF}{\ensuremath{\mathfrak{L}}\xspace} 
\newcommand{\LT}{\SF}

\newcommand{\SFf}{\ensuremath{\mathfrak{Lf}}\xspace} 
\newcommand{\LTf}{\SFf}

\usepackage{changebar}

\begin{document}
\title{Planning and Synthesis Under Assumptions}
\author{
Benjamin Aminof\\
  JKU Linz and TU Wien\\
  Austria\\
  aminof@forsyte.at
\And
  Giuseppe De Giacomo \\
  Sapienza Univ.\ Roma \\
  Rome, Italy \\
  degiacomo@dis.uniroma1.it
\And
  Aniello Murano\\
  Univ.\ Federico II \\
  Naples, Italy \\
  murano@unina.it
\And
  Sasha Rubin \\
  Univ.\ Federico II \\
  Naples, Italy \\
  rubin@unina.it
}

\maketitle

\begin{abstract}
  In Reasoning about Action and Planning, one synthesizes the agent
  plan by taking advantage of the assumption on how the environment
  works (that is, one exploits the environment's effects, its
  fairness, its trajectory constraints).  In this paper we study this
  form of synthesis in detail. We consider assumptions as constraints
  on the possible strategies that the environment can have in order to
  respond to the agent's actions.  Such constraints may be given in the
  form of a planning domain (or action theory), as linear-time
  formulas over infinite or finite runs, or as a combination of the
  two. We argue though that not all
  assumption specifications are meaningful: they need to be
  consistent, which means that there must exist an environment
  strategy fulfilling the assumption in spite of the agent
  actions. For such assumptions, we study how to do synthesis/planning
  for agent goals, ranging from a classical reachability to goal on
  traces specified in \LTL and \LTLf/\LDLf, characterizing the
  problem both mathematically and algorithmically.
\end{abstract}


\section{Introduction}

Reasoning about actions and planning concern the representation of a dynamic
system. This representation consists of a description of the interaction
between an agent and its environment and aims at enabling reasoning and
deliberation on the possible course of action for the agent \cite{Reit01}.
Planning in fully observable nondeterministic domains (FOND), say in Planning Domain Definition Language (PDDL),
\cite{GhNT04,GeBo13} exemplifies the standard methodology for expressing
dynamic systems: it represents the world using finitely many \emph{fluents}
under the control of the \emph{environment} and a finitely many \emph{actions}
under the control of the \emph{agent}. Using these two elements a model of the
dynamics of world is given. Agent goals, e.g., reachability objectives, or,
say, temporally extended objectives written in \LTL
\cite{BacchusK00,CTMBM17,DR-IJCAI18}, are expressed over such models in terms of
such fluents and actions.

An important observation is that, in devising plans, the agent takes
advantage of such a representation of the world. Such a representation
corresponds to knowledge that the agent has of the world. In other
words, the agent \emph{assumes} that the world works in a certain
way, and \emph{exploits such an assumption in devising its plans}.
A question immediately comes to mind: 
 \begin{quote}
 \emph{Which kinds of environment assumptions can the agent make?}
 \end{quote}
Obviously the planning domain itself (including the initial state) with its
preconditions and effects is such an assumption. That is, as long as the agent 
sticks to its preconditions, the environment acts as described by the domain. 
So, the agent can exploit the effect of its actions in order to reach a certain goal 
{(state of affairs)}.
{Another common assumption is to assume the domain is \emph{fair}, i.e.,
so-called \emph{fair FOND}~\cite{DaTV99,PistoreT01,Cimatti03,CTMBM17,DIppolitoRS18}.  In this case
the agent can exploit not only the effects, but also the guarantee that by
continuing to execute an action {from a given state} the environment will
eventually respond with all its possible nondeterministic effects.}\footnote{There are two notions of fairness in planning. One stems from the fact that nondeterminism is resolved stochastically. The other is a logical notion analogous to that used in the formal-methods literature. These two notions coincide in the context reachability goals~\cite{DIppolitoRS18}, but diverge with more general LTL goals~\cite{Pnueli:STOC83,Pnueli:IC93}. In this paper, we focus on the logical notion.}
More recently \cite{DBLP:conf/ijcai/BonetG15,DBLP:conf/ijcai/BonetGGR17} 
trajectory constraints over the domain, expressed in \LTL, have been proposed
to model general restrictions on the possible environment behavior.
But is any kind of \LTL formula on the fluents and actions of the domain a
possible trajectory constraint for the environment?
The answer is obviously not! To see this, consider a formula expressing that
eventually a certain possible action must actually be performed (the agent may
decide not to do it). But then 
\begin{quote}
\emph{Which trajectory constraints are suitable as assumptions in a given
domain?} 
\end{quote}
Focusing on \LTL, the question can be rephrased as: 
\begin{quote}
\emph{Can any {linear-time} specification be used as an assumption for the environment?}
\end{quote}
We can summarize these questions, ultimately, by asking:
\begin{quote}
\emph{What is an environment assumption?}
\end{quote}

This is what we investigate in this paper. We take the view that environment
assumptions are ways to talk about \emph{the set of strategies the environment can enact}. Moreover, 
the plan for the goal, i.e., the agent strategy for fulfilling
the goal, need only fulfill the goal against the strategies of the
environment {from the given set of environment strategies}. 
We formalize this insight and {\emph{define}} synthesis/planning under assumptions and
the relationship between the two in a general linear-time setting. {In particular, our definitions
only allow linear-time properties to be assumptions if the environment can enforce them.}
In doing this \emph{we answer the above questions}.

We also concretize the study and express goals and assumptions in \LTL, automata over infinite  words (deterministic parity word automata)~\cite{ALG02}, as well as formalisms over finite traces, i.e., \LTLf/\LDLf~\cite{DegVa13,DR-IJCAI18} and finite word automata. 
This allows us to study \emph{how to solve} synthesis/planning under assumptions problems. One may
think that the natural way to solve such synthesis problems is to have 
the agent synthesize a strategy for the implication 
\[ Assumption \limp Goal\] 
where both $Assumption$ and $Goals$ are expressed, say, in \LTL. 
A first problem with such an implication is that the agent should not devise
strategies that make $Assumption$ false, because in this case the agent would
lose its model of the world without \emph{necessarily} fulfilling its $Goal$.  
{This undesirable situation is avoided by our very notion of environment assumption.}
{A second issue is this:}
\begin{quote}
\emph{Does synthesis/planning under assumptions amount to
synthesizing for the above implication?} 
\end{quote}
{We show that this is not the case. 
Note that an agent that synthesizes for the implication is too pessimistic: the agent, having chosen a candidate agent strategy, considers as possible all environment strategies that satisfy $Assumption$ against the specific candidate strategy it is analyzing. But, in this way the agent gives too much power to the environment, since, in fact, the environment does not know the agent's chosen strategy. On the other hand, surprisingly, we show that if there is an agent strategy fulfilling $Goal$ under $Assumption$, then also there exists one that indeed enforces the implication.
}
Thus, even if the implication 
\emph{cannot be used for characterizing the problem of
synthesis/planning under assumptions}, it \emph{can be used to solve it}.
Exploiting this result, we give techniques to solve synthesis/planning under
assumptions, and study the worst case complexity of the problems when goals and assumptions are
expressed {in the logics and automata mentioned above.}


\section{Synthesis and Linear-time specifications} \label{sec:prelims}

\emph{Synthesis} is the problem of producing a module that satisfies a given property no matter how the environment behaves ~\cite{PnueliR89}.
Synthesis can be thought of in the terminology of games. Let $\AP$ be a finite set of Boolean variables (also called atoms), and assume it is partitioned into two sets: $A$, those controllable by the agent, and $E$, those controllable by the environment. 
Let $\A = 2^A$ be the set of \emph{actions} and $\E = 2^E$ the set of \emph{environment states} (note the symmetry: we 
think of $\A$ as a set of actions that are compactly
represented as assignments of the variables in $A$). The game consists of infinitely many phases. 
In each phase of the game, both players assign values to their variables, with the environment going first. These assignments 
are given by \emph{strategies}: an agent strategy $\sigma_\ag:\E^+ \to \A$ and an environment strategy $\sigma_\env:\A^* \to \E$. The resulting infinite sequence of assignments is denoted $\pi_{\sigma_\ag,\sigma_\env}$.\footnote{Formally, 
we say that $\pi = \pi_0 \pi_1 \cdots$ complies with $\sigma_\ag$ if $\sigma_\ag((\pi_0 \cap E) \cdots (\pi_k \cap E)) = \pi_k \cap A$ for all $k$; and we say that $\pi$ complies with $\sigma_\env$ if $\sigma_\env((\pi_0 \cap A) \cdots (\pi_k \cap A)) = \pi_{k+1} \cap E$ for all $k$. Then $\pi_{\sigma_\ag,\sigma_\env}$ is defined to be the unique infinite trace that complies with both strategies.} 

{In classic synthesis the agent is trying to ensure that the produced sequence satisfies a given linear-time property. In what follows \emph{we write \SF to denote a generic formalism for defining linear-time properties.} Thus, the reader may substitute their favorite formalism for \SF, e.g., one can take \SF to be linear temporal logic, or deterministic parity automata. We use logical notation throughout. For instance, when $\phi$ refers to a logical formula, then $\phi_1 \wedge \phi_2$ refers to conjunction of formulas, but when $\phi$ refers to an automaton then $\phi_1 \wedge \phi_2$ refers to intersection of automata. If $\phi \in \LT$ write $[[\phi]] \in (2^\AP)^\omega$ for the set it defines. For instance, if $\phi \in \LTL$ then $[[\phi]]$ is the set of infinite sequences that satisfy $\phi$, but when $\phi$ is an automaton operating on infinite sequences, then $[[\phi]]$ is the set of infinite sequences accepted by the automaton. Moreover, in both cases we say that the sequence \emph{satisfies} $\phi$.
}

We say that \emph{$\sigma_\ag$ realizes $\phi$ (written $\sigma_\ag \sat \phi$)} if $\forall \sigma_\env. \pi_{\sigma_\ag,\sigma_\env} \in [[\phi]]$, i.e., if no matter which strategy the environment uses, the resulting sequence satisfies $\phi$. Similarly, we say that \emph{$\sigma_\env$ realizes $\phi$ (written $\sigma_\env \sat \phi$)} if $\forall \sigma_\ag. \pi_{\sigma_\ag,\sigma_\env} \in [[\phi]]$. We write $Str_\env(\phi)$ (resp. $Str_\ag(\phi)$) for the set of environment (resp. agent) strategies that realize $\phi$, and in case this set is non-empty we say that \emph{$\phi$ is environment (resp. agent) realizable}. We write $Str_\env$ (resp. $Str_\ag$) for the set of all environment (resp. agent) strategies.

\emph{Solving $\SF$ environment- (resp. agent-) synthesis} asks, given $\phi \in \SF$ to decide if $\phi$ is environment- (resp. agent-) 
realizable, and to return such a finite-state strategy (if one exists). 
{In other words, realizability is the {recognition problem} associated to synthesis.}
We now recall two concrete specification formalisms $\SF$, namely \LTL (linear temporal logic) and \DPW (deterministic parity word automata), and then state that results about solving \LTL/\DPW synthesis.

\subsubsection{Linear temporal Logic (LTL)} \label{sec:prelims:LTL}

	 \emph{Formulas of $\LTL(\AP)$}, or simply \emph{$\LTL$}, are generated by the following grammar:
	\[\varphi \!::=\! p \!\mid\! \varphi \vee \varphi \!\mid\! \neg \varphi \!\mid\!  \nextX \! \varphi \!\mid \! \varphi \until \varphi\]
	where $p \in \AP$. The \emph{size $|\varphi|$} of a formula $\varphi$ is the number of symbols in it.
	\LTL formulas are interpreted over infinite sequences $\pi \in (2^{\AP})^\omega$. Define the satisfaction relation
	$\models$ as follows:
	\begin{enumerate*}
	\item[] $(\pi,n) \models p$ iff $p \in \pi_n$;
	\item[]  $(\pi,n) \models \varphi_1 \vee \varphi_2$ iff $(\pi,n) \models \varphi_i$ for some $i \in \{1,2\}$;
	\item[] 	$(\pi,n) \models \neg \varphi$ iff it is not the case that $(\pi,n) \models \varphi$;
	\item[]   $(\pi,n) \models \nextX \varphi$ iff $(\pi,n+1) \models \varphi$;
	\item[]  $(\pi,n) \models \varphi_1 \until \varphi_2$ iff there exists $i \geq n$ such that $(\pi,i) \models \varphi_2$ and for all $i \leq j < n$, $(\pi,j) \models \varphi_1$.
	\end{enumerate*}
	Write $\pi \models \varphi$ if $(\pi,0) \models \varphi$ and say that $\pi$ \emph{satisfies} $\varphi$ and $\pi$ is a \emph{model} of $\varphi$. 
	An \LTL formula $\varphi$ \emph{defines} the set $[[\pi]] \doteq \{\pi \in (2^\AP)^\omega: \pi \models \varphi\}$.
	We use the usual abbreviations, $\varphi \limp \varphi' \doteq \neg \varphi \vee \varphi'$, $\true := p \vee \neg p$, $\false \doteq \neg \true$, 
	$\eventually \varphi \doteq \true \until \varphi$,
	$\always \varphi \doteq \neg \eventually \neg \varphi$.  
	Write $Bool(\AP)$ for the set of Boolean formulas over $\AP$. We remark that every result in this paper that mentions \LTL also holds for \LDL (linear dynamic logic)~\cite{Var11,Eisner:2006tu}. 
	
\subsubsection{Deterministic Parity Word Automata (DPW)}

A \emph{\DPW} over $\AP$ is a tuple $M = (Q,q_{in},T,col)$ where $Q$ is a finite set of \emph{states}, 
$q_{in} \in Q$ is an \emph{initial state}, $T:Q \times 2^{\AP} \to Q$ is the \emph{transition function}, and $col:Q \to \mathbb{Z}$ 
is the \emph{coloring}. The \emph{run} $\rho$  of $M$ on the \emph{input} word $x_0 x_1 x_2 \cdots \in (2^{\AP})^\omega$ 
is the infinite sequence of transitions $(q_0,x_0,q_1) (q_1,x_1,q_2) (q_2,x_2,q_3) \cdots$ such that $q_0 = q_{in}$. 
A run is \emph{successful} if the largest color occurring infinitely often is even. 
In this case, we say that the input word is \emph{accepted}. The \DPW $M$ 
\emph{defines} the set $[[M]]$ consisting of all input words it accepts. The \emph{size of $M$}, written $|M|$, is the cardinality of $Q$.
The \emph{number of colors} of $M$ is the cardinality of $col(Q)$. 

{\DPW{s} are effectively closed under Boolean operations, see e.g., \cite{ALG02}:
\begin{lemma} \label{lem:DPW Boolean combinations} 
Let $M_i$ be \DPW with $n_i$ states and $c_i$ colors, respectively.
\begin{enumerate} 
 \item One can effectively form a \DPW with $n_1$ states and $c_1$ colors for the complement of $M_1$.
 \item One can effectively form a \DPW with with $O(n_1 n_2 d^2d!)$ many states and $O(d)$ many colors, where $d = c_1 + c_2$, for the disjunction $M_1 \lor M_2$.
\end{enumerate}
\end{lemma}
Thus, e.g., from \DPW $M_1,M_2$ one can build a DPW for $M_1 \limp M_2$ whose number of states is $O(n_1 n_2 d^2d!)$ and whose number of colors is $O(d)$.

Every \LTL formula $\varphi$ can be translated into an equivalent \DPW $M$, i.e., $[[\varphi]] = [[M]]$, see e.g. \cite{DBLP:conf/banff/Vardi95,DBLP:journals/lmcs/Piterman07}.
Moreover, the cost of this translation and the size of $M$ are at most doubly exponential in the size of $\varphi$, 
and the number of colors of $M$ is at most singly exponential in the size of $\varphi$.

Here is a summary of the complexity of solving synthesis:
\begin{theorem}[Solving Synthesis] \label{fact:synthesis} \hspace{0cm}
\begin{enumerate} 
\item Solving \LTL environment (resp. agent) synthesis is $2$\exptime-complete~\cite{PnueliR89}.
\item Solving \DPW environment (resp. agent) synthesis is \ptime 
in the size of the automaton and $\exptime$ in the number of its colors~\cite{PnueliR89,finkbeiner2016synthesis}. 
\end{enumerate}
\end{theorem}


\section{Synthesis under Assumptions}

In this section we give core definitions of environment assumptions and synthesis under such assumptions.
Intuitively, the assumptions are used to select the environment strategies that the agent considers possible, i.e., although the agent does not know 
the particular environment strategy it will encounter, it knows that it comes from such a set.
We begin in the abstract, and then move to declarative specifications. Unless explicitly specified, we assume fixed sets $E$ and $A$ of environment and agent atoms.

{Here are the main definitions of this paper:}
\begin{definition}[Environment Assumptions -- abstract] \label{def:abstract:assumptions}
We call any {non-empty} set $\Omega \subseteq Str_\env$ of environment strategies an \emph{environment assumption}.
\end{definition}

Informally, the set $\Omega$ represents the set of environment strategies that the agent considers possible.
\begin{definition}[Agent Goals -- abstract] 
We call any set $\Gamma$ of traces an \emph{agent goal}. 
\end{definition}

\begin{definition}[Synthesis under assumptions -- abstract] \label{def:abstract:sua}
Let $\Omega$ be an environment assumption and $\Gamma$ an agent goal. 
We say that an agent strategy $\sigma_\ag$ \emph{realizes $\Gamma$ assuming $\Omega$} if 
\[\forall \sigma_\env \in \Omega. \, \pi_{\sigma_\ag,\sigma_\env} \in \Gamma\]  
\end{definition}

\begin{remark}[On the non-emptiness of $\Omega$]
Note that the requirement that $\Omega$ be non-empty is a consistency requirement; if it were empty then there would be no $\pi_{\sigma_\ag,\sigma_\env}$ to test for membership in $\Gamma$ and so synthesis under assumptions would trivialize and all agent strategies would realize all goals. 
\end{remark}

For the rest of this paper we will specify agent goals and environment assumptions as linear-time properties. \textbf{In particular, we assume that $\SF$ is a formalism for specifying linear-time properties over $\AP$, e.g., $\SF = \LTL$ or $\SF = \DPW$.  }

How should $\omega \in \SF$ determine an assumption $\Omega$? In general, $\omega$ talks about the interaction between the agent and the environment.
However, we want that the agent can be guaranteed that whatever it does the resulting play satisfies $\omega$. Thus, a given $\omega$ induces the set $\Omega$ consisting of {all} environment strategies $\sigma_\env$ such that for all agent strategies $\sigma_\ag$ the resulting trace satisfies $\omega$. In particular, for $\Omega$ to be non-empty (as required for it to be an environment assumption) we must have that $\omega$ is environment realizable. This justifies the following definitions.

\begin{definition}[Synthesis under Assumptions -- linear-time] \label{dfn:LT:assumptions} \hspace{0cm}
\begin{enumerate}
 \item  We call $\omega \in \SF$ an \emph{environment assumption} if it is environment realizable.
 \item We call any $\gamma \in \SF$ an \emph{agent goal}.
 \item An \emph{$\SF$ synthesis under assumptions problem} is a tuple $P = (E,A,\omega,\gamma)$ where $\omega \in \SF$ is an environment assumption and $\gamma \in \SF$ is an agent goal.
 \item We say that an agent strategy $\sigma_\ag$ \emph{realizes $\gamma$ assuming $\omega$}, or that it \emph{solves} $P$, if 
 $\forall \sigma_\env \sat \omega. \, \pi_{\sigma_\ag,\sigma_\env} \models \gamma$. 
 \item The corresponding decision problem is to decide, given $P$, if there is an agent strategy solving $P$.
\end{enumerate}
\end{definition}
For instance, \emph{solving \LTL synthesis under assumptions} means, given $P = (E,A,\omega,\gamma)$ with environment assumption 
$\omega \in \LTL(E \cup A)$ and agent goal $\gamma \in \LTL(E \cup A)$, to decide if there is an agent strategy solving $P$, and to return such a finite-state strategy (if one exists).
We remark that solving \LTL synthesis under assumptions is not immediate; we will provide algorithms in the next section. For now, we point out that 
\emph{deciding whether $\omega$ is an environment assumption amounts to checking if $\omega$ is environment realizable}, itself a problem that can be solved by {known results (i.e., Theorem~\ref{fact:synthesis}).}
\begin{theorem}
\begin{enumerate}
 \item Deciding if an $\LTL$ formula is an environment assumption is $2$\exptime-complete. 
 \item Deciding if a $\DPW$ is an environment assumption is in \ptime in the size of the \DPW and exponential in its number of colors.
\end{enumerate}
\end{theorem}

We illustrate such notions with some examples.
\begin{example} \label{ex:asmp}
\begin{enumerate}
   \item The set $\Omega = Str_\env$, definable in \LTL by the formula $\omega \doteq \true$, is an environment assumption. It captures the situation that the agent assumes that the environment will use any of the strategies in $Str_\env$.
   
   \item In robot-action planning problems, typical environment assumptions encode the physical space, e.g., ``if robot is in Room 1 and does action $Move$ then in the next step it can only be in Rooms 1 or 4''. 
   The set $\Omega$ of environment strategies that realize these properties is an environment assumption, 
   definable in \LTL by a conjunction of formulas of the form $\always((R_1 \wedge Move) \limp \nextX (R_1 \vee R_4))$. 
   We will generalize this example by showing that the set of environment strategies in a planning domain $D$ can be viewed as an environment assumption definable in \LTL. 
   \end{enumerate}
\end{example}


\section{Solving Synthesis under Assumptions}

In this section we show how to solve synthesis under assumptions when the environment assumptions and agent goals are given in \LTL or by \DPW.
{The general idea is to reduce synthesis under assumptions to ordinary synthesis, i.e., synthesis of the implication $\omega \limp \gamma$. 
Although correct, understanding why it is correct is not immediate.}

\begin{lemma}
Let $\omega \in \LT$ be an environment assumption and $\gamma \in \LT$ an agent goal. Then, 
every agent strategy that realizes $\omega \limp \gamma$ also realizes $\gamma$ assuming $\omega$. 
\end{lemma}
\begin{proof}
Let $\sigma_\ag$ be an agent strategy realizing $\omega \limp \gamma$ (a). To show that $\sigma_\ag$ realizes $\gamma$ assuming $\omega$ let 
$\sigma_\env$ be an environment strategy realizing $\omega$ (b). Now consider the trace $\pi = \pi_{\sigma_\ag,\sigma_\env}$. We must show that $\pi$ satisfies $\gamma$. By (a) $\pi$ satisfies $\omega \limp \gamma$ and by (b) $\pi$ satisfies $\omega$. 
\end{proof}

We now observe that the converse is not true. Consider $A \doteq \{x\}$ and $E \doteq  \{y\}$, and let $\omega \doteq y \limp x$ and $\gamma \doteq y \limp \neg x$. First note that $\omega$ is an environment assumption formula (indeed, the environment can realize $\omega$ by playing $\neg y$ at the first step). Moreover, every environment strategy realizing $\omega$ begins by playing $\neg y$ (since otherwise the agent could play $\neg x$ on its first turn and falsify $\omega$). Thus, every agent strategy realizes $\gamma$ assuming $\omega$ (since the environment's first move is to play $\neg y$ which makes $\gamma$ true no matter what the agent does). On the other hand, not every agent strategy realizes $\omega \limp \gamma$ 
(indeed, the strategy which plays $x$ on its first turn fails to satisfy the implication on the trace in which the environment plays $y$ on its first turn).
In spite of the failure of the converse, the realizability problems are inter-reducible:{\footnote{For all reasonable expressions $\omega$, e.g., that define Borel sets~\cite{Mar75}.}}
\begin{theorem} \label{thm:det}
Suppose $\omega \in \LT$ is an environment assumption.
 The following are equivalent:
 \begin{enumerate}
  \item There is an agent strategy realizing $\omega \limp \gamma$.
  \item There is an agent strategy realizing $\gamma$ assuming $\omega$.
 \end{enumerate}
\end{theorem}

\begin{proof}
The previous lemma gives us $1 \rightarrow 2$. For the converse, suppose $1$ does not hold, i.e., $\omega \limp \gamma$ is not agent-realizable. 
Now, an immediate consequence of Martin's Borel Determinacy Theorem~\cite{Mar75} is that for every $\phi$ in any reasonable specification formalism (including all the ones mentioned in this paper), $\phi$ is not agent realizable iff $\neg \phi$ is environment realizable. 
Thus, $\neg (\omega \limp \gamma)$ is environment-realizable, i.e., 
$\exists \sigma_\env \forall \sigma_\ag. \pi_{\sigma_\ag,\sigma_\env} \models \omega \wedge \neg \gamma$. Note in particular that $\sigma_\env$ realizes $\omega$, i.e., $\sigma_\env \sat \omega$. Now, suppose for a contradiction that $2$ holds, and take $\sigma_\ag$ realizing $\gamma$ assuming $\omega$. Then by definition of realizability under assumptions and using the fact that $\sigma_\env \sat \omega$ we have that $\pi_{\sigma_\ag,\sigma_\env} \models \gamma$. On the other hand, we have already seen that $\pi_{\sigma_\ag,\sigma_\env} \models \neg \gamma$, a contradiction.
\end{proof}

Moreover, we see that one can actually extract a strategy solving 
synthesis by assumptions simply by extracting a strategy for solving
the implication $\omega \limp \gamma$, which itself can be done by
known results, i.e., for \LTL use Theorem~\ref{fact:synthesis} (part 1), and for \DPW use Lemma~\ref{lem:DPW Boolean combinations} and Theorem~\ref{fact:synthesis} (part 2).
\begin{theorem}\label{thm:solving:SUA}
\begin{enumerate} 
\item Solving \LTL synthesis under assumptions is $2$\exptime-complete.
\item Solving \DPW synthesis under assumptions is in \ptime in the size of the automata and in \exptime in the number of colors of the automata.
\end{enumerate}
\end{theorem}


\section{Planning under Assumptions} \label{sec:Planning Under Assumptions}
In this section we define planning under assumptions,
{that is synthesis wrt a domain\footnote{Domains can be thought of as compact representations of the arenas in games on graphs~\cite{ALG02}. The player chooses actions, also represented compactly, and the environment resolves the nondeterminism. In addition, not every action needs to be available in every vertex of the arena.}.}
We begin with a representation of fully-observable non-deterministic (FOND) domains~\cite{GhNT04,GeBo13}. Our representation considers actions symmetrically to fluents, i.e., as assignments to certain variables.

A \emph{domain} $D = (E,A,I,Pre,\Delta)$ consists of:
\begin{itemize} 
 \item a non-empty set $E$ of \emph{environment} Boolean variables, also called \emph{fluents}; 
 the elements of $\E = 2^E$ are called \emph{environment states},
 \item a non-empty set $A$ (disjoint from $E$) of \emph{action} Boolean variables; the elements of  $\A = 2^A$ are called \emph{actions}, 
 \item a non-empty set $I \subseteq \E$ of \emph{initial environment states},
 \item a relation $Pre \subseteq \E \times \A$ of \emph{available actions} such that for every $s \in \E$ there is an $a \in \A$ with $(s,a) \in Pre$ (we say that $a$ is \emph{available} in $s$), and 
 \item a relation $\Delta \subseteq \E \times \A \times \E$ such that $(s,a,t) \in \Delta$ implies that $(s,a) \in Pre$.
\end{itemize}

As is customary in planning and reasoning about actions, we assume domains are represented compactly by tuples $(E,A,init,pre,\delta)$ where 
$init \in Bool(E)$, 
$pre \in Bool(E \cup A)$, 
and $\delta \in Bool(E \cup A \cup E')$ (here $E' \doteq \{e' : e \in E\}$). This data induces the domain $(E,A,I,Pre,\Delta)$ where 
\begin{enumerate}
 \item $s \in I$ iff $s \models init$,
 \item $(s,a) \in Pre$ iff $s \cup a \models pre$, 
 \item $(s,a,t) \in \Delta$ iff $s \cup a \cup \{e' : e \in t\} \models \delta$.
\end{enumerate}

We emphasize that when measuring the size of $D$ we use this compact representation:
\begin{definition} 
The \emph{size of $D$}, written $|D|$, is $|E| + |A| + |init| + |pre| + |\delta|$. 
\end{definition}

We remark that in PDDL action preconditions are declared using $\texttt{:precondition}$, conditional effects using the $\texttt{when}$ operator, and nondeterministic outcomes using the $\texttt{oneof}$ operator (note that we code actions with action variables).

\begin{example}[Universal Domain]
Given $E$ and $A$ define the \emph{universal} domain $U = (E,A,I,Pre,\Delta)$ where $I \doteq \E$, $Pre \doteq \E \times \A$ and $\Delta \doteq \E \times \A \times \E$. 
\end{example}

We now define the set of environment strategies induced by a domain. We do this by describing a property $\omega_D$, that itself can be 
represented in \LTL and \DPW, as shown below.

\begin{definition} \label{dfn:domain-env} 
Fix a domain $D$.
Define a property $\omega_D$ (over atoms $E \cup A$) as consisting of all traces $\pi = \pi_0 \pi_1 \ldots$ such that 
\begin{enumerate} 
\item $\pi_0 \in I$ and 
\item for all $n \geq 1$, if $\pi_i \cap A$ is available in $\pi_i \cap E$ for every $i \in [0,n-1]$ then $(\pi_{n-1} \cap E, \pi_{n-1} \cap A, \pi_{n} \cap E) \in \Delta$.
\end{enumerate}

\end{definition}
Observe that $\omega_D$ is an environment assumption since, by the definition of domain, whenever an action is available in a state there is at least one possible successor state. 
Intuitively, an environment strategy $\sigma_\env:\A^* \to \E$ is in $Str_\env(\omega_D)$ if i) its first move is to pick an initial environment state, and ii) thereafter, if the current action $a$ is available in the current environment state $x$ (and the same holds in all earlier steps) then the next environment state $y \in \E$ is constrained so that $(x,a,y) \in \Delta$. Notice that $\sigma_\env$ is unconstrained the moment $a$ is not available in $x$, e.g., in PDDL these would be actions for which the preconditions are not satisfied. Intuitively, this means that it is in the interest of the agent to play available actions because otherwise the agent can't rely {on the fact that the trace comes from the domain.}

\begin{remark}
The reader may be wondering why the above definition does not say i') $\pi_0 \in I$ and ii') for all $n \geq 1$, $(\pi_{n-1} \cap E, \pi_{n-1} \cap A, \pi_{n} \cap E) \in \Delta$. Consider the linear-time property $\omega'_D$ consisting of traces $\pi$ satisfying i' and ii'. Observe that, in general, $\omega'_D$ is not environment realizable. Indeed, condition ii' implies that $\pi_n \cap A$ is available in $\pi_n \cap E$. However, no environment strategy can force the agent to play an available action.
\end{remark}

We now observe that one can express $\omega_D$ in \LTL.

\begin{lemma} \label{lem:omegaD:LTL}
For every domain $D$ there is an \LTL formula equivalent to $\omega_D$. 
Furthermore, the size of the \LTL formula is linear in the size of $D$. 
\end{lemma}

To see this, say domain $D = (E,A,I,Pre,\Delta)$ is represented compactly by $(E,A,init,pre,\delta)$. For the \LTL formula, let $\delta'$ be the $\LTL(E \cup A)$ formula formed from the formula $\delta \in Bool(E \cup A \cup E')$ by replacing every term of the form $e'$ by $\nextX e$. Note that $(\pi,n) \models \delta'$ iff $(\pi_n \cap E, \pi_n \cap A, \pi_{n+1} \cap E) \in \Delta$. 
The promised $\LTL(E \cup A)$ formula is 
\[ init \wedge (\always \delta' \vee \delta' \until \neg pre).\]

{One can also express $\omega_D$ directly by a \DPW.}
\begin{lemma} \label{lem:omegaD:DPW}
For every domain $D$ there is a \DPW $M_D$ equivalent to $\omega_D$. 
Furthermore, the size of the \DPW is at most exponential in the size of $D$ and has two colors.
\end{lemma}

To do this we define the \DPW directly rather than translate the \LTL formula (which would give a double exponential bound).  
Define the \DPW $M_D \doteq (Q,q_{in},T,col)$ over $E \cup A$ as follows. 
Introduce fresh symbols $q_{in}, q_{+},q_{-}$. Let $q_{in}$ be the initial state. Define $Q \doteq \{q_{in},q_{+},q_{-}\} \cup (\E \times \A)$.
Define $col(q_{-}) = 1$, and $col(q) = 0$ for all $q \neq q_{-}$.
For all $e,e' \in \E, a,a' \in \A$ the transitions are given in Table~\ref{tab:trans}. Intuitively, on reading the input $e' \cup a'$ the \DPW goes to the rejecting sink $q_{-}$ if $\Delta$ (resp. $I$) is not respected, 
it goes to the accepting sink $q_{+}$ if $\Delta$ (resp. $I$) is respected but $Pre$ is not, 
and otherwise it continues (and accepts).
  \begin{table}[h!]
 \begin{tabular}{llll}
$q_{in}$ & $\xrightarrow{e' \cup a'}$& $q_{-}$   & if $e' \not \in I$\\
$q_{in}$ & $\xrightarrow{e' \cup a'}$& $(e',a')$ & if $e' \in I$ and $(e',a') \in Pre$ \\
$q_{in}$ & $\xrightarrow{e' \cup a'}$& $q_{+}$   & if $e' \in I$ and $(e',a') \not \in Pre$\\
$(e,a)$  & $\xrightarrow{e' \cup a'}$& $q_{-}$   & if $(e,a,e') \not \in \Delta$\\
$(e,a)$  & $\xrightarrow{e' \cup a'}$& $(e',a')$ & if $(e,a,e') \in \Delta$ and $(e',a') \in Pre$\\
$(e,a)$  & $\xrightarrow{e' \cup a'}$& $q_{+}$   & if $(e,a,e') \in \Delta$ and $(e',a') \not \in Pre$\\
$q_{-}$  & $\xrightarrow{e' \cup a'}$& $ q_{-}$  &\\
$q_{+}$  & $\xrightarrow{e' \cup a'}$& $ q_{+}$  &\\
\end{tabular}
  \caption{Transitions for \DPW for $\omega_D$}
  \label{tab:trans}
\end{table}

\begin{definition}\label{dfn:planning:assumption}
Let $D$ be a domain. 
\begin{itemize} \item A set $\Omega \subseteq Str_\env$ is an \emph{environment assumption for the domain $D$} if $Str_\env(\omega_D) \cap \Omega$ is non-empty. 
 \item $\omega \in \SF$ is an \emph{environment assumption for the domain $D$} if $Str_\env(\omega_D) \cap Str_\env(\omega)$ is non-empty, i.e.,  if $Str_\env(\omega)$ is an environment assumption for the domain $D$.
\end{itemize}

\end{definition}

{We illustrate the notion with some examples.}
\begin{example} \label{ex} \hspace{0cm}
\begin{enumerate} 
 
 \item $\omega \doteq \true$ is an environment assumption for $D$ since $\omega_D \wedge \omega \equiv \omega_D$ is environment realizable.
 
 \item Let $\omega_{D,fair}$ denote the following property: $\pi \in \omega_{D,fair}$ iff for all $(s,a) \in Pre$, if there are infinitely many $n$ such that $s = \pi_n \cap E$ and $a = \pi_n \cap A$, then for every $t \in \E$ with $(s,a,t) \in \Delta$ there are infinitely many $n$ such that $s = \pi_n \cap E, a = \pi_n \cap A$ and $t = \pi_{n+1} \cap E$. In words, this says that if a state-action pair occurs infinitely often, then infinitely often this is followed by every possible effect.
 
 Note that $\omega_{D,fair}$ is an environment assumption for domain $D$ since, e.g., the strategy that
 resolves the effects in a round-robin way realizes $\omega_D \wedge \omega_{D,fair}$. Note that $\omega_{D,fair}$ is definable in \LTL by a formula of size exponential in $D$: 
 \[
\bigwedge_{s \in \E} \bigwedge_{a \in \A} (\always \eventually (s \wedge a)  \limp \bigwedge_{s': (s,a,s') \in \Delta} \always \eventually (s \wedge a \wedge \nextX s')). 
\]

\item In planning, trajectory constraints, e.g., expressed in LTL,  have been introduced for expressing temporally extended goals ~\cite{BacchusK00,DBLP:journals/ai/GereviniHLSD09}.
More recently, especially in the context of generalized planning, they have been used to describe restrictions on the environment as well \cite{DBLP:conf/ijcai/BonetG15,DeGiacomoMRS16,DBLP:conf/ijcai/BonetGGR17}. 
However, not all trajectory constraints $\omega$ can be used as assumptions.  In fact, Definition~\ref{dfn:planning:assumption}, which says that a formula $\omega$ is an environment assumption for the domain $D$ if $\omega_D \wedge \omega$ is environment realizable, characterizes those formulas that can serve as trajectory constraints.
\end{enumerate}
\end{example}

We can check if $\omega \in \LTL$ is an environment assumption for $D$ by converting it to a DPW $M_\omega$, converting $D$ into the \DPW $M_D$ (as above), and then checking if the \DPW $M_D \wedge M_\omega$ is environment realizable. Hence we have:
\begin{theorem} 
\begin{enumerate}
 \item Deciding if an $\LTL$ formula $\omega$ is an environment assumption for the domain $D$ is $2$\exptime-complete. Moreover, it can be solved in 
 \exptime in the size of $D$ and $2$\exptime in the size of $\omega$. 
 \item Deciding if a \DPW $\omega$ is an environment assumption for the domain $D$ is in \exptime. Moreover, it can be solved in \exptime in the size 
 of $D$ and \ptime in the size of $\omega$ and \exptime in the number of colors of $\omega$.
\end{enumerate}
\end{theorem}
For the lower bound take $D \doteq U$ to be the universal domain and apply the lower bound from Theorem~\ref{fact:synthesis}.

{Now we turn to planning under assumptions.}
\begin{definition}[Planning under Assumptions -- abstract] \label{def:PUA}\hspace{0cm}
\begin{enumerate}
\item A \emph{planning under assumptions problem $P$} is a tuple $((D,\Omega),\Gamma)$ where 
 \begin{itemize} 
  \item $D$ is a domain,
  \item $\Omega \subseteq Str_\env$ is an environment assumption for $D$, and 
  \item $\Gamma$ is an agent goal.
 \end{itemize}
 \item  We say that an agent strategy $\sigma_\ag$ \emph{solves} $P$ if 
 \[\forall \sigma_\env \in Str_\env(\omega_D) \cap \Omega. \, \pi_{\sigma_\ag,\sigma_\env} \in \Gamma\]  
\end{enumerate}
\end{definition}

We can instantiate this definition to environment assumptions and agent goals definable in $\SF$.
\begin{definition}[Planning under Assumptions -- linear-time] \hspace{0em}
\begin{enumerate}
 \item An \emph{$\SF$ planning under assumptions problem} is a tuple $P = ((D,\omega),\gamma)$ where $\omega \in \SF$ is an environment assumption for $D$ and $\gamma \in \SF$ is an agent goal.
 \item We say that an agent strategy $\sigma_\ag$ \emph{realizes $\gamma$ assuming $\omega$}, or that it \emph{solves} $P$, if 
 \[\forall \sigma_\env \sat (\omega_D \wedge \omega). \, \pi_{\sigma_\ag,\sigma_\env} \models \gamma\] 
\end{enumerate}
\end{definition}

The corresponding decision problem asks, given an \LT planning under assumptions problem $P$ to decide whether there is an agent strategy that solves $P$. For instance, \emph{\LTL planning under assumptions} asks, given $P = ((D,\omega),\gamma)$ with $\omega,\gamma \in \LTL$, 
to decide if there is an agent strategy that solves $P$, and to return such a finite-state strategy (if one exists). Similar definitions apply to \DPW planning under assumptions, etc.

{It turns out that {virtually all  forms of planning {(with linear-time temporally extended goals)}} in the literature are special cases of planning under \LTL assumptions, i.e., {the set of strategies that solve a given planning problem are exactly the set of strategies that solve the corresponding planning under assumptions problem.} 
In the following, $\mathit{Goal} \in Bool(E \cup A)$, and 
$\mathit{Exec}$ is the \LTL formula $\always \bigwedge_{a \in A} (a \limp \mathit{pre_a})$ expressing that if an action is done then its precondition holds.}

\begin{example}

 \begin{enumerate}
\item FOND planning with reachability goals~\cite{Rintanen:ICAPS04} corresponds to \LTL planning under assumptions with $\omega \doteq \true$ and {$\gamma \doteq \mathit{Exec} \land \eventually Goal$.}

\item FOND planning with \LTL (temporally extended) goals $\gamma$ 
\cite{BacchusK00,PistoreT01,CTMBM17}.
corresponds to \LTL planning under assumptions with $\omega \doteq \true$ {and goal $\mathit{Exec} \land \gamma$.}

\item FOND planning with \LTL trajectory constraints $\omega$ and \LTL (temporally extended) goals $\gamma$ \cite{DBLP:conf/ijcai/BonetG15,DeGiacomoMRS16,DBLP:conf/ijcai/BonetGGR17} corresponds to \LTL planning under assumptions {with assumptions $\omega$ and goal $\mathit{Exec} \land \gamma$.}

\item {Fair FOND planning} with reachability goals~\cite{DaTV99,GeBo13,DIppolitoRS18} corresponds to planning under assumptions with 
$\omega \doteq \omega_{D,fair}$ and $\gamma \doteq \textit{Exec} \land \eventually Goal$.

\item {Fair FOND planning} with (temporally extended) goals $\gamma$ as defined in~\cite{DBLP:conf/ijcai/PatriziLG13,CTMBM17} corresponds to planning under assumptions with $\omega \doteq \omega_{D,fair}$ and goal $\mathit{Exec} \land \gamma$.

\item Obviously adding \LTL trajectory constraints $\omega_{tc}$ to {fair FOND planning} with (temporally extended) goals  corresponds to planning under assumptions with $\omega \doteq \omega_{D,fair} \wedge \omega_{tc}$ {and goal $\mathit{Exec} \land \gamma$.} 
\end{enumerate}

\end{example}

{We also observe that the Fair FOND planning problems just mentioned can be captured by \LTL planning under assumptions since $\omega_{D,fair}$ can be written in \LTL (see Example~\ref{ex}).


\section{Translating between planning and synthesis}
In this section we ask the question if there is a {fundamental} difference between synthesis and planning in our setting (i.e., assumptions and goals given as linear-time properties). We answer by observing that there are translations between them. The next two results follow immediately from the definitions:

\begin{theorem}[Synthesis to Planning] \label{prop:synthesis to planning}
Let $(E,A,\omega,\gamma)$ be a synthesis under Assumptions problem, and let 
$P = ((U,\omega),\gamma)$ be the corresponding 
Planning under Assumptions problem where $U$ is the universal domain. Then, for every agent strategy $\sigma_\ag$ we have that 
$\sigma_\ag$ solves $P$ iff $\sigma_\ag$ realizes $\gamma$ assuming $\omega$.
\end{theorem}

\begin{theorem}[Planning to Synthesis] \label{prop:planning to synthesis}
Let $D = (E,A,I,Pre,\Delta)$ be a domain and let $P = ((D,\omega),\gamma)$ be a Planning under Assumptions problem. 
Let $(E,A,\omega_D \wedge \omega,\gamma)$ be the corresponding Synthesis under Assumptions problem. Then, for every agent strategy $\sigma_\ag$ we have that 
$\sigma_\ag$ solves $P$ iff $\sigma_\ag$ realizes $\gamma$ assuming $\omega_D \wedge \omega$.
\end{theorem}

Thus, we can solve \LTL planning under assumptions by reducing to \LTL synthesis under assumptions, which itself can be solved by known results (i.e., Theorem~\ref{fact:synthesis}):
\begin{corollary}\label{thm:solving:PUA:LTL:combined}
 Solving \LTL planning under assumptions is $2$\exptime-complete.
\end{corollary}
However, this does not distinguish the complexity measured in the size of the domain from that in the size of the assumption and goal formulas. 
{We take this up next.}


\section{Solving Planning under Assumptions}

In this section we show how to solve Planning under Assumptions for concrete specification languages \SF, i.e., \SF = \LTL and \SF = \DPW. We measure the complexity in two different ways: we fix the domain $D$ and measure the complexity with respect to the size of the formulas or automata for the environment assumption and the agent goal, this is called \emph{goal/assumption complexity}; 
and  we fix the formulas/automata and measure the complexity with respect to the size of the domain, this is called the \emph{domain complexity}. \footnote{Formally, if $C$ is a complexity class, we say that \emph{goal/assumption complexity is in $C$} if for every domain $D_0$ the complexity of deciding if there is an agent strategy solving $P = ((D_0,\omega),\gamma)$, is in $C$. A similar definition holds for domain complexity. Also, we say that the \emph{goal/assumption complexity is $C$-hard} if there exists a domain $D_0$ such that the problem of deciding if there is an agent strategy solving $P = ((D_0,\omega),\gamma)$, is $C$-hard.}

We begin with \SF = \DPW and consider the following algorithm:
Given $P = ((D,\omega),\gamma)$ in which $\omega$ is represented by a \DPW $M_\omega$ and $\gamma$ is represented by a \DPW $M_\gamma$, perform the following steps:

 \begin{tabbing}
===\===\===\===\=\+\kill
\textbf{Alg 1. Solving \DPW planning under assumptions}\\
Given domain $D$, assumption $M_\omega$, goal $M_\gamma$.\\
1:\>Form \DPW $M_D$ equivalent to $\omega_D$.\\
2:\>Form \DPW $M$ for $(M_D \wedge M_\omega) \limp M_\gamma$.\\
3:\>Solve the parity game on $M$.
\end{tabbing}

The first step results in a \DPW whose size is exponential in the size of $D$ and with a constant number of colors (Lemma~\ref{lem:omegaD:DPW}). 
The second step results in a \DPW whose size is polynomial in the number of states of the \DPW{s} involved (i.e., $M_D,M_\omega$ and $M_\gamma$), and exponential in the number of their colors (Lemma~\ref{lem:DPW Boolean combinations}).
For the third step, the think of the \DPW $M$ as a parity game: play starts in the initial state, and at each step, if $q$ is the current state of $M$, first the environment picks $s \in \E$ and then the agent picks an action $a \in \A$, i.e., an evaluation of the action variables. The subsequent step starts in the state of $M$ resulting from taking the unique transition from $q$ labeled $s \cup a$. This produces a run of the \DPW which the agent is trying to ensure is successful (i.e., the largest color occurring infinitely often is even). 

Formally, we say that an agent strategy $\sigma_\ag$ is \emph{winning} if 
for every environment strategy $\sigma_\env$, the unique run of the \DPW on input word $\pi_{\sigma_\ag,\sigma_\env}$ is successful. 
Deciding if the a player has a winning strategy, and returning a finite-state strategy (it one exists), is called \emph{solving} the game. 
Parity games can be solved in 
time polynomial in the size of $M$ and exponential in the number of colors of $M$~\cite{ALG02}.\footnote{Better algorithms are known, e.g.~\cite{CaludeJKL017}, but are not helpful for this paper.}  

{The analysis of the above algorithm shows the following.}
\begin{theorem} \label{thm:solving:PUA:DPW} \hspace{0em}
\begin{enumerate} \item The domain complexity of solving \DPW planning under assumptions is in \exptime. 
\item The goal/assumption complexity of solving \DPW planning under assumptions is in \ptime in their sizes and \exptime in the number of their colors.
\end{enumerate}
\end{theorem}

{Moreover,} by converting \LTL formulas to \DPW with exponentially many colors and double-exponential many states \cite{DBLP:conf/banff/Vardi95,DBLP:journals/lmcs/Piterman07}, we get the upper bounds in the following:
\begin{theorem} \label{thm:solving:PUA:LTL}
\begin{enumerate}
\item The domain complexity of solving \LTL planning under assumptions is \exptime-complete.
\item The goal/assumption complexity of solving \LTL planning under assumptions is $2$\exptime-complete. 
\end{enumerate}
\end{theorem}
{For the matching lower-bounds, we have that}
the domain complexity is \exptime-hard follows from the fact that planning with reachability goals and no assumptions is \exptime-hard~\cite{Rintanen:ICAPS04}; to see that the goal/assumption complexity is $2$\exptime-hard note that \LTL synthesis, known to be $2$\exptime-hard~\cite{PnueliR89,rosner1992modular}, is a special case (take $\omega \doteq \true$ and $D$ to be the universal domain). 

{Similarly, one can apply this technique to solving Fair \LTL planning under assumptions. The exact complexity, however, is open. See the conclusion for a discussion.}


\section{Focusing on finite traces}
In this section we revisit the definitions and results in case that assumptions and goals are expressed as linear-time properties over \emph{finite} traces. There are two reasons to do this. 
First,  in AI and CS applications executions of interest are often finite~\cite{DegVa13}.
Second, the algorithms presented for the infinite-sequence case involve complex constructions on automata/games that are notoriously hard to optimize~\cite{DFogartyKVW13}. Thus, we will not simply reduce the finite-trace case to the infinite-trace case~\cite{GMM14}. We begin by carefully defining the setting.

\subsubsection{Synthesis and linear-time specifications over finite traces}
We define synthesis over finite traces in a similar way to the infinite-trace case, {cf.~\cite{DegVa15,Camacho:KR18}.}
The main difference is that agent strategies 
$\sigma_\ag:\E^+ \to \A$ can be partial. This represents the situation that the agent stops the play. Environment strategies $\sigma_\env:\A^* \to \E$ are total (as before). Thus, the resulting play $\pi_{\sigma_\ag,\sigma_\env}$ may be finite, if the agent chooses to stop, as well as infinite.\footnote{Formally,  $\pi_{\sigma_\ag,\sigma_\env}$ is redefined to be the longest trace (it may be finite or infinite) that complies with both strategies.} Objectives may be expressed in general specification formalisms \SFf for finite traces, e.g., \SFf = \LTLf (\LTL over finite traces\footnote{All our results for \LTLf also hold for linear-dynamic logic over finite traces (\LDLf)~\cite{DegVa13}. 
}), \SFf = \DFA (deterministic finite word automata). For $\phi \in \LTf$, we overload notation and write $[[\phi]]$ for the set of finite traces $\phi$ defines.

We now define realizability in the finite-trace case:
\begin{definition}
Let $\phi \in \LTf$.
\begin{enumerate} 
\item We say that \emph{$\sigma_\ag$ realizes $\phi$ (written $\sigma_\ag \sat \phi$)} if $\forall \sigma_\env. \left(\pi_{\sigma_\ag,\sigma_\env} \text{ is finite and } \pi_{\sigma_\ag,\sigma_\env} \in [[\phi]]\right)$.  
\item We say that \emph{$\sigma_\env$ realizes $\phi$ (written $\sigma_\env \sat \phi$)} if $\forall \sigma_\ag. \left(\text{if } \pi_{\sigma_\ag,\sigma_\env} \text{ is finite, then } \pi_{\sigma_\ag,\sigma_\env}\in [[\phi]]\right)$. 
\end{enumerate}
\end{definition}
The asymmetry in the definition results from the fact that stopping is controlled by the agent.

Duality still holds, and is easier to prove since it amounts to determinacy of reachability games~\cite{ALG02}:
\begin{lemma}[Duality] \label{lem:duality:finite}
For every $\phi \in \SFf$ we have that $\phi$ is not agent realizable iff $\neg \phi$ is environment realizable.
\end{lemma}

\subsubsection{Linear temporal logic on finite traces (\LTLf)} The logic $\LTLf$ has the same syntax as $\LTL$ but is interpreted on finite traces $\pi \in (2^\AP)^+$. Formally, for $n \leq len(\pi)$ (the length of $\pi$) we only reinterpret the temporal operators:
\begin{itemize}
	\item  $(\pi,n) \models \nextX \varphi$ iff $n < len(\pi)$ and $(\pi,n+1) \models \varphi$;
	\item  $(\pi,n) \models \varphi_1 \until \varphi_2$ iff 
	there exists $i$ with $n \leq i \leq len(\pi)$ such that $(\pi,i) \models \varphi_2$ and for all $i \leq j < n$, $(\pi,j) \models \varphi_1$.
\end{itemize}
Let $\tilde{\nextX}$ denote the dual of $\nextX$, i.e., $\tilde{\nextX} \doteq \neg \nextX \neg \varphi$. Semantically we have that
\begin{itemize}
	\item  $(\pi,n) \models \tilde{\nextX} \varphi$ iff $n < len(\pi)$ implies $(\pi,n+1) \models \varphi$.
\end{itemize}

\subsubsection{Deterministic finite automata (\DFA)} 

A \DFA over $\AP$ is a tuple $M = (Q,q_{in},T,F)$ which is like a \DPW except that $col$ is replaced by a set $F \subseteq Q$ of final states. The run on a finite input trace $\pi \in (2^\AP)^*$ is successful if it ends in a final state. We recall that \DFA are closed under Boolean operations using classic algorithms (e.g., see \cite{DBLP:conf/banff/Vardi95}).
Also, \LTLf formulas $\varphi$ (and also \LDLf formulas) can be effectively translated into \DFA. This is done in three classic simple steps that highlight the power of the automata-theoretic approach: convert $\varphi$ to an alternating automaton (poly), then into a nondeterministic finite automaton (exp), and then into a \DFA (exp). These steps are outlined in detail in, e.g., \cite{DegVa13}.

\subsubsection{Solving Synthesis over finite traces} \LTf agent synthesis is the problem, given $\phi \in \LTf$, of deciding if the agent can realize $\phi$. Now, solving \DFA agent synthesis is \ptime-complete: it amounts to solving a reachability game on the given \DFA $M$, which can be done with an algorithm that captures how close the agent is to a final state, i.e., a least-fixpoint of the operation. Finally, to solve \LTLf agent synthesis first translate the \LTLf formula to a \DFA and then run the fixpoint algorithm  (also, \LTLf agent synthesis is $2$\exptime-complete)~\cite{DegVa15}.

Note that, by Duality, solving \LTLf environment realizability and solving \LTLf agent realizability are inter-reducible (and thus the former is also $2$\exptime-complete). Thus, to decide if $\phi$ is environment realizable we simply negate the answer to whether $\neg \phi$ is agent realizable. However, to extract an environment strategy, one solves the dual safety game.

\subsubsection{Synthesis under assumptions}
We say that $\omega \in \LTf$ is an \emph{environment assumption} if $\omega$ is environment realizable. Solving 
\LTf synthesis under assumptions means to decide if there is an agent strategy $\sigma_\ag$ such that
\[\forall \sigma_\env \sat \omega. \left(\pi_{\sigma_\ag,\sigma_\env} \text{ is finite and } \pi_{\sigma_\ag,\sigma_\env} \models \gamma\right).\]

We now consider the case that $\LTf = \LTLf$. Checking if $\omega \in \LTLf$ is an environment assumption is, by definition, 
the problem of deciding if $\omega$ is environment realizable, as just discussed. Hence we can state the following:
\begin{theorem} \hspace{0em}
\begin{enumerate}
 \item Deciding if an $\LTLf$ formula $\omega$ is an environment assumption is $2$\exptime-complete. 
 \item Deciding if a $\DFA$ $\omega$ is an environment assumption is \ptime-compete (cf.~\cite{ALG02}).
\end{enumerate}
\end{theorem}

Turning to \LTLf synthesis under assumptions  we have 
that synthesis under assumptions and synthesis of the implication are equivalent. Indeed, as before, the key point is the duality which we have in 
Lemma~\ref{lem:duality:finite}:

\begin{theorem} \label{thm:det:finite}
Suppose $\omega \in \LTf$ is an environment assumption.
 The following are equivalent:
 \begin{enumerate}
  \item There is an agent strategy realizing $\omega \limp \gamma$.
  \item There is an agent strategy realizing $\gamma$ assuming $\omega$.
 \end{enumerate}
\end{theorem}

Hence to solve synthesis under assumptions we simply solve agent synthesis for the implication. Hence we have:
\begin{theorem}\label{thm:solving:SUA:finite} \hspace{0em}
\begin{enumerate} 
\item Solving \LTLf synthesis under assumptions is $2$\exptime-complete.
\item Solving \DFA synthesis under assumptions is \ptime-complete.
\end{enumerate}
\end{theorem}

\subsubsection{Planning under assumptions}
Planning and fair planning have recently been studied for \LTLf goals~\cite{DR-IJCAI18,Camacho:KR18,Camacho:ICAPS18}. Here we define and study how to add environment assumptions.

Recall that we represent a planning domain $D$ by the linear-time property $\omega_D$ (Definition~\ref{dfn:domain-env}) which itself was defined as those infinite traces satisfying two conditions. The exact same conditions determine a set of finite traces, also denoted $\omega_D$. Moreover, 
this $\omega_D$ is equivalent to an \LTLf formula of size linear in $D$ and a \DFA of size at most exponential in $D$. 
To see this, replace $\nextX$ by $\tilde{\nextX}$ in the \LTL formula from Lemma~\ref{lem:omegaD:LTL}. That is, 
let $\delta''$ be the $\LTLf$ formula formed from $\delta$ by replacing every term of the form $e'$ by $\tilde{\nextX} e$. Note that if 
$n < len(\pi)$ then $(\pi,n) \models \delta''$ iff $(\pi_n \cap E, \pi_n \cap A, \pi_{n+1} \cap E) \in \Delta$, and if $n = len(\pi)$ then $(\pi,n) \models \delta''$ iff $(\pi_n \cap E, \pi_n \cap A) \in Pre$. 
The promised $\LTLf(E \cup A)$ formula is 
$init \wedge (\always \delta'' \vee \delta'' \until \neg pre)$. 
Also, similar to the \DPW before there is a \DFA of size at most exponential in the size of $D$ equivalent to $\omega_D$. 
To see this, take the \DPW $M_D \doteq (Q,q_{in},T,col)$ from Lemma~\ref{lem:omegaD:DPW} and instead of $col$ define the set 
of final states to be the set $col^{-1}(0)$.

As before, say that $\omega \in \LTf$ is an \emph{environment assumption for the domain $D$} if $\omega_D \wedge \omega$ is environment realizable. 
Define an \emph{\LTf planning under assumptions problem} to be a tuple $P = ((D,\omega),\gamma)$ with $\omega,\gamma \in \LTf$ 
such that $\omega$ is an environment assumption for $D$. To decide if $\omega \in \LTLf/\DFA$ is an environment assumption for $D$ 
we use the next algorithm:
 \begin{tabbing}
===\===\===\===\=\+\kill
\textbf{Alg 2. Deciding if $\omega$ is an environment assumption for $D$}\\
Given domain $D$, and \DFA $M_\omega$.\\
1:\>Convert $D$ into a \DFA $M_D$ equivalent to $\omega_D$.\\
2:\>Form the \DFA $M$ for $(M_D \wedge M_\omega)$.\\
3:\>Decide if $M$ is environment realizable.
\end{tabbing}
Further, if $\omega$ is given as an \LTLf formula, first convert it to a \DFA $M_\omega$ and then run the algorithm. We then have:
\begin{theorem}  \hspace{0em}
\begin{enumerate}
 \item Deciding if $\LTLf$ formula $\omega$ is an environment assumption for the domain $D$ is $2$\exptime-complete. Moreover, it can be solved in 
 \exptime in the size of $D$ and $2$\exptime in the size of $\omega$. 
 \item Deciding if $\DFA$ $\omega$ is an environment assumption for the domain $D$ is in \exptime. Moreover, it can be solved in \exptime in the size 
 of $D$ and \ptime in the size of $\omega$.
\end{enumerate}
\end{theorem}

\subsubsection{Solving Planning under Assumptions}
As before, there are simple translations between \LTf planning under assumptions and \LTf synthesis under assumptions. And again, solving \LTLf planning under assumptions via such a translation is not fine enough to analyze the complexity in the domain vs the goal/assumption. 
To solve \DFA/\LTLf planning under assumptions use the following simple algorithm:

 \begin{tabbing}
===\===\===\===\=\+\kill
\textbf{Alg 3. Solving \DFA planning under assumptions}\\
Given domain $D$, assumption $M_\omega$, goal $M_\gamma$.\\
1:\>Convert $D$ into a \DFA $M_D$ equivalent to $\omega_D$.\\
2:\>Form the \DFA $M$ for $(M_D \wedge M_\omega) \limp M_\gamma$.\\
3:\>Solve the reachability game on \DFA $M$.
\end{tabbing}

Further, if $\omega$ is given as an \LTLf formula, first convert it to a \DFA $M_\omega$ and then run the algorithm. This gives the upper bounds in the following:
\begin{theorem} \label{thm:solving:PUA:DFW} \hspace{0em} 
\begin{enumerate} 
\item The domain complexity of solving \DFA (resp. \LTLf) planning under assumptions is \exptime-complete.
\item The goal/assumption complexity of solving \DFA (resp. \LTLf) planning under assumptions is \ptime-complete (resp. $2$\exptime-complete). 
\end{enumerate}
\end{theorem}
For the lower bounds, setting $\omega \doteq \true$ results in FOND with reachability goals, known to be \exptime-hard~\cite{Rintanen:ICAPS04}; and additionally taking the domain $D$ to be the universal domain results in \DFA (resp. \LTLf) synthesis, known to be \ptime-hard~\cite{ALG02} (resp. $2$\exptime-hard~\cite{DegVa15}).

Finally, if $P = ((D,\omega),\gamma)$ is an \LTLf planning under assumptions problem, say that $\sigma_\ag$ \emph{fairly solves} $P$ if for every $\sigma_\env \sat \omega_D \wedge \omega$ we have that if $\pi_{\sigma_\ag,\sigma_\env} \in [[\omega_{D,fair}]]$ then $\pi_{\sigma_\ag,\sigma_\env}$ is finite and satisfies $\gamma$ (here $\omega_{D,fair}$ from Example~\ref{ex} is defined so that it now also includes all finite traces). We remark that Alg $2$ applies unchanged. 
{However, to solve the fair \LTLf planning problem, we do not know a better way, in general, than translating the problem into one over infinite traces and applying the techniques from the previous section.}


\section{Conclusion and Outlook}

While we illustrate synthesis and planning under assumptions expressed in
linear-time specifications, our definitions immediately apply to assumptions
expressed in branching-time specifications, e.g., $\CTLS$, $\mu$-calculus, and
tree automata.  As future work, it is of great interest to study synthesis
under assumptions in the branching time setting so as to devise restrictions on
\emph{possible agent behaviors} with certain guarantees, e.g., remain in an
area from where the agent can enforce the ability to reach the recharging doc,
whenever it needs to, in the spirit of~\cite{LagoPT02}.

Although our work is in the context of reasoning about actions and planning, we
expect it can also provide insights to verification and to multi-agent systems.
In particular, the undesirable drawback of the agent being able to falsify an
assumption when synthesizing $Assumption \limp Goal$ is well known, and it has
been observed that it can be overcome when the $Assumption$ is environment
realizable~\cite{DBLP:journals/tosem/DIppolitoBPU13,DBLP:journals/acta/BrenguierRS17}.
Our Theorem~\ref{thm:det} provides the principle for such a solution.
Interestingly, various degrees of cooperation to fulfill assumptions among
adversarial agents has been considered, e.g.,
\cite{CH07,BEK15,DBLP:journals/acta/BrenguierRS17} and we believe that a work
like present one is needed to establish similar principled foundations.

Turning to the multi-agent setting, there, agents in a common environment
interact with each other and may have their own objectives. Thus, it makes
sense to model agents not as hostile to each other, but as rational, i.e.,
agents that act to achieve their own objectives. \emph{Rational
synthesis}~\cite{KPV14} (as compared to classic synthesis) further requires
that the strategy profile chosen by the agents is in equilibrium (various
notions of equilibrium may be used). It would be interesting to investigate
rational synthesis under environment assumptions, in the sense that all agents
also make use of their own assumptions about their common environment. We
believe that considering assumptions as sets of strategies rather than sets of
traces will serve as a clarifying framework also for the multi-agent setting.

{Finally, there are a number of open questions regarding the computational complexity 
of solving synthesis/planning under assumptions, i.e., what is the exact complexity of {Fair \LTL/\LTLf planning under assumptions}? what is the assumption complexity of \LTL/\LTLf synthesis under assumptions? Here, the \emph{assumption} complexity is the complexity of the problem assuming the domain and goal are fixed, and the only input to the problem is the assumption formula/automaton.}


\bibliography{main}
\bibliographystyle{aaai}

\end{document}